\documentclass[11pt,letterpaper]{article}
\usepackage{times,mathptmx}
\DeclareMathAlphabet{\mathcal}{OMS}{cmsy}{m}{n}
\usepackage{fullpage}
\usepackage{amsmath,amsfonts,amssymb,amsthm,amscd}
\usepackage{hyperref,cite,url}
\hypersetup{pdfpagemode=UseNone,pdfstartview=}
\newcommand{\vs}{\vspace{1.5mm}}
\theoremstyle{plain} % plain
\newtheorem{theorem}{Theorem}[section]
\newtheorem{lemma}[theorem]{Lemma}

\theoremstyle{definition} % definition
\newtheorem{definition}{Definition}[section]

\theoremstyle{remark} % remark

\newcommand{\G}{\mathbb{G}}
\newcommand{\Z}{\mathbb{Z}}

\newcommand{\mc}[1]{\mathcal{#1}}
\newcommand{\tb}[1]{\textbf{#1}}
\newcommand{\lb}{\linebreak[0]}

\title{On the Analysis of the Revocable-Storage Identity-Based
Encryption Scheme}

\author{
    Kwangsu Lee\footnote{Sejong University, Seoul, Korea.
    Email: \texttt{kwangsu@sejong.ac.kr}.}
}

\date{}

\begin{document}

\maketitle

\begin{abstract}
Cloud computing can provide a flexible way to effectively share data among
multiple users since it can overcome the time and location constraints of
computing resource usage. However, the users of cloud computing are still
reluctant to share sensitive data to a cloud server since the cloud server
should be treated as an untrusted entity. In order to support secure and
efficient data sharing in cloud computing environment, Wei et al. recently
extended the concept of identity-based encryption (IBE) to support key
revocation and ciphertext update functionalities, and proposed a
revocable-storage identity-based encryption (RS-IBE) scheme. In this paper,
we show that the RS-IBE scheme of Wei et al. does not satisfy the
correctness property of RS-IBE. In addition, we propose a method to modify
the existing RS-IBE scheme to be a correct and secure scheme.
\end{abstract}

\vs \noindent {\bf Keywords:} Cloud computing, Identity-based encryption,
Revocable-storage, Ciphertext update.

\section{Introduction}

Cloud computing is a new paradigm of computing system that provides computing
resources such as computing power or data storage according to the need of
users. The advantage of cloud computing is that cloud service users can use
their computing resources as a service with low cost at any time from
anywhere through the Internet. Many technology companies provide various
types of cloud services. The main difference between traditional server
computing and cloud computing is that a cloud service provider can no longer
be regarded as a trusted entity. In other words, the cloud service provider
should be treated as a honest-but-curious adversary.

A typical application of cloud computing is to securely share data among a
large number of users. In this system, the data confidentiality should be
provided because the cloud service provider is no longer a trusted entity. In
addition, if a user's credential is expired or the user's private key is
compromised, then a proper revocation method should be provided to handle
this user. Furthermore, even if a revoked user tries to access past data
stored in the cloud server through collusion attacks, the security of data
should be guaranteed. Therefore, a secure data sharing system in the cloud
server should consider various security issues described above.

Recently, a revocable-storage identity-based encryption (RS-IBE) scheme for
secure data sharing in cloud storage was proposed by Wei et al.
\cite{WeiLH18}. The basic idea of this RS-IBE scheme is to modify an
identity-based encryption (IBE) scheme to additionally supports the key
revocation and ciphertext update functionalities. In particularly, they used
the IBE scheme of Waters for the underlying IBE scheme and the tree-based key
revocation scheme of Boldyreva et al. \cite{BoldyrevaGK08} for key
revocation. Additionally, they modified their scheme to support efficient
ciphertext update by following the idea of forward-secure cryptographic
systems.

In this paper, we show that there is a serious problem in the RS-IBE scheme
of Wei et al. That is, if a ciphertext generated at time $T$ is updated to
another ciphertext with time $T+1$ by the ciphertext update algorithm, then
this updated ciphertext with time $T+1$ cannot be decrypted by using a
decryption key with time $T+1$. The reason of this decryption failure problem
is that the decryption algorithm uses a random value which is different from
the random value used to encrypt a message if a ciphertext is updated. A more
detailed explanation of this problem is given in the later part of this work.
To remedy this problem, we propose a method to modify the RS-IBE scheme of
Wei et al. to be a secure scheme without the decryption failure problem.

The organization of the paper is as follows. In Section 2, we review the
definition and scheme of RS-IBE proposed by Wei et al. In Section 3, we point
out that there is a correctness problem in Wei et al.'s RS-IBE scheme. In
Section 4, we propose a method to solve this problem by. Finally, we conclude
the paper in Section 5.

\section{Review of the RS-IBE Scheme}

In this section, we review the definition of RS-IBE including the correctness
property and describe the RS-IBE scheme of Wei et al.

\subsection{Revocable-Storage Identity-Based Encryption}

\begin{definition}[Revocable-Storage Identity-Based Encryption]
A revocable-storage identity-based encryption (RS-IBE) scheme consists of
eight algorithms \tb{Setup}, \tb{GenKey}, \tb{UpdateKey}, \tb{DeriveDK},
\tb{Encrypt}, \tb{UpdateCT}, \tb{Decrypt}, and \tb{Revoke}, which are defined
as follows:
\begin{description}
\item \tb{Setup}($1^\lambda, N_{max}, T_{max}$). The setup algorithm takes
    as input a security parameter $1^{\lambda}$, the maximum number of
    users $N_{max}$, and the total number of time periods $T_{max}$. It
    outputs a master key $MK$ and public parameters $PP$.

\item \tb{GenKey}($ID, MK, PP$). The key generation algorithm takes as
    input an identity $ID$, the master key $MK$, and the public parameters
    $PP$. It outputs a private key $SK_{ID}$.

\item \tb{UpdateKey}($T, RL, MK, ST, PP$). The key update algorithm takes
    as input update time $T$, a revocation list $RL$, the master key $MK$,
    a state $ST$, and the public parameters $PP$. It outputs a key update
    $KU_T$.

\item \tb{DeriveDK}($SK_{ID}, KU_T, PP$). The decryption key derivation
    algorithm takes as input a private key $SK_{ID}$, a key update $KU_T$,
    and the public parameters $PP$. It outputs a decryption key
    $DK_{ID,T}$.

\item \tb{Encrypt}($ID, T, M, PP$). The encryption algorithm takes as input
    an identity $ID$, time $T$, a message $M$, and the public parameters
    $PP$. It outputs a ciphertext $CT_{ID,T}$.

\item \tb{UpdateCT}($CT_{ID,T}, T', PP$). The ciphertext update algorithm
    takes as a ciphertext $CT_{ID,T}$, update time $T'$, and the public
    parameters $PP$. It outputs an updated ciphertext $CT_{ID,T'}$.

\item \tb{Decrypt}($CT_{ID,T'}, DK_{ID,T}, PP$). The decryption algorithm
    takes as input a ciphertext $CT_{ID,T'}$, a decryption key $DK_{ID,T}$,
    and the public parameters $PP$. It outputs an encrypted message $M$ or
    $\perp$.

\item \tb{Revoke}($ID, T, RL, ST, PP$). The revocation algorithm takes as
    input an identity $ID$, revoked time $T$, a revocation list $RL$, and a
    state $ST$. It outputs an updated revocation list $RL$.
\end{description}
The correctness property of RS-IBE is defined as follows: For all $MK, PP$
generated by \tb{Setup}, any $SK_{ID}$ generated by $\tb{GenKey}(ID, MK, PP)$
for any $ID$, any $KU_T$ generated by $\tb{UpdateKey}(T, RL, MK, ST, PP)$ for
any $T, RL$, $CT_{ID,T'}$ generated by $\tb{Encrypt}(ID, T', M, PP)$ for any
$ID, T', M$, if $ID$ is not revoked at time $T$ in $RL$, then it is required
that $DK_{ID,T}$ can be derived by $\tb{DeriveKey}(SK_{ID}, KU_T, PP)$ and
\begin{itemize}
\item If $T' \leq T$, then $\tb{Decrypt} (CT_{ID,T'}, DK_{ID,T}, PP) = M$.

\item If $T' > T$, then $\tb{Decrypt} (CT_{ID,T'}, DK_{ID,T}, PP) = \perp$.
\end{itemize}
Additionally, it is required that the ciphertext distribution of
$\tb{UpdateCT} (CT_{ID,T}, T', PP)$ is statistically equal to that of
$\tb{Encrypt} (ID, T', M, PP)$.
\end{definition}

\subsection{Wei et al.'s RS-IBE Construction}

To provide key revocation functionality, the RS-IBE scheme of Wei et al.
\cite{WeiLH18} follows the binary tree-based broadcast encryption method
proposed by Boldyreva et al. \cite{BoldyrevaGK08}. Let $\mc{BT}$ be a binary
tree for handling key revocation. A user is randomly assigned to a leaf node
in this $\mc{BT}$. At this time, the private key $SK_{ID}$ of a user with an
identity $ID$ is associated with the set of nodes defined by
$\tb{Path}(x_{ID})$ which is the set of path nodes from the root node to the
leaf node $x_{ID}$, and a key update $KU_T$ at time $T$ is associated with
the set of covering nodes defined by $\tb{KUNodes}(\mc{BT}, RL, T)$ which is
the set of nodes that covers all non-revoked leaf nodes at time $T$. If the
leaf node (or the private key) of a user $ID$ is not revoked at time $T$,
then there is a common node $x$ satisfying $\tilde{x} = \tb{Path}(x_{ID})
\cap \tb{KUNodes}(\mc{BT}, RL, T)$. The decryption of a ciphertext at time
$T$ can be possible by using the private key element and key update element
corresponding to the node $\tilde{x}$. For the detailed description of
$\tb{KUNodes}(\mc{BT}, RL, T)$, please refer the work of Boldyreva et al.
\cite{BoldyrevaGK08}.

To provide ciphertext update functionality, this RS-IBE scheme uses the
binary tree idea of Canetti et al. \cite{CanettiHK03} used to build
forward-secure encryption schemes. Note that the binary tree idea for time
management was widely used in other RS-ABE schemes \cite{SahaiSW12,
LeeCLPY13,LeeCLPY17,Lee16}. Let $\mc{ET}$ be a binary tree to handle time in
a ciphertext. In this case, each time is sequentially allocated to a leaf
node in $\mc{ET}$ from left to right. In this case, $\tb{CTNodes}(\mc{ET},
T)$ is defined as $\tb{RightSibling}(\tb{Path}(v_T)) \setminus
\tb{Path}(\tb{Parent}(v_T)) \cup \{ v_T \}$ where $\tb{RightSibling}(S)$ is a
set of $\tb{RightChild}( \tb{Parent}(v))$ of any node $v \in S$
\cite{SahaiSW12,LeeCLPY13}. Note that Wei et al. wrongly defined
$\tb{CTNodes}(\mc{ET}, T) = \{ v | \tb{Parent}(v) \in \tb{Path}(v_T) \text{
and } v \notin \tb{Path}(v_T) \} \cup \{ v_T \}$ because this (wrongly
defined) set can include the left child node of $\tb{Path}(v_T)$, which will
allow access to the past time node. To support ciphertext update, a
ciphertext is constructed to have ciphertext elements associated with
$\tb{CTNodes}(\mc{ET}, T)$. The main property of $\tb{CTNodes}$ is that if $T
< T'$, a ciphertext with $\tb{CTNodes}(\mc{ET}, T)$ can be updated to a
ciphertext with $\tb{CTNodes}(\mc{ET}, T')$ because for any node $v' \in
\tb{CTNodes}(\mc{ET}, T')$ there is a node $v''$ that matches to
$\tb{CTNodes}(\mc{ET}, T) \cap \tb{Path}(v')$ and the ciphertext component
for $v''$ can be delegated to be a ciphertext for $v'$. For other properties
of $\tb{CTNodes}$, please refer the work of Sahai et al. \cite{SahaiSW12}.

The RS-IBE scheme of Wei et al. is described as follows:

\begin{description}
\item [\tb{Setup}($1^\lambda, N_{max}, T_{max}$):] Let $\lambda$ be the
    security parameter, $N_{max} = 2^n$ be the maximum number of users, and
    $T_{max} = 2^{\ell}$ be the total number of time periods. It chooses a
    bilinear group $(p, \G, \G_T, e)$ with a prime order $p$.
    It selects random $g, g_2 \in \G$ and $\alpha \in \Z_p^*$, and sets
    $g_1 = g^{\alpha}$. It also chooses random $u_0, u_1, \ldots, u_n, h_0,
    h_1, \ldots, h_\ell \in \G$ and defines $F_u(ID) = u_0 \prod_{i=1}^{n}
    u_i^{ID[i]}$, $F_h(T) = h_0 \prod_{j=1}^{\ell} h_j^{T[j]}$ where
    $ID[i]$ and $T[i]$ are the $i$th bit of $ID$ and $T$ respectively.
    It sets a binary tree $\mc{BT}$ with $N_{max}$ number of leaf nodes and
    sets a revocation list $RL = \emptyset$, a state $ST = \mc{BT}$. It
    outputs a master key $MK = g_2^{\alpha}$, and public parameters $PP =
    \big( (p, \G, \G_T, e), g, g_1, g_2, \{ u_i \}_{i=0}^{n}, \{ h_i
    \}_{i=0}^{\ell} \big)$.

\item [\tb{GenKey}($ID, MK, ST, PP$):] It assigns $ID$ to a leaf node
    $x_{ID} \in \mc{BT}$. For each node $x \in \tb{Path}(x_{ID})$, it
    performs:
    1) It fetches $g_{x,0}$ from the node $x$. If $g_{x,0}$ is not defined
    before, then it chooses random $g_{x,0} \in \G$ and stores the pair
    $(g_{x,0}, g_{x,1} = g_2 \cdot g_{x,0}^{-1})$ in the node $x$.
    2) It chooses random $r_{x,0} \in \Z_p^*$ and obtains $SK_{ID,x} =
    \big( K_{x,0} = g_{x,0}^{\alpha} F_u(ID)^{r_{x,0}}, K_{x,1} =
    g^{r_{x,0}} \big)$.
    Finally, it outputs a private key $SK_{ID} = \big( \{ (x, SK_{ID,x})
    \}_{x \in \tb{Path}(x_{ID})} \big)$ and an updated $ST = \mc{BT}$.

\item [\tb{UpdateKey}($T, RL, MK, ST, PP$):] For each node $x \in
    \tb{KUNodes}(\mc{BT}, RL, T)$, it performs:
    1) It fetches $g_{x,1}$ from the node $x$. If $g_{x,1}$ is not defined,
    then it sets the value similar to the key generation algorithm.
    2) It chooses random $r_{x,1} \in \Z_p^*$ and obtains $KU_{T,x} = \big(
    U_0 = g_{x,1}^{\alpha} F_h(T)^{r_{x,1}}, U_1 = g^{r_{x,1}} \big)$.
    Finally, it outputs a key update $KU_T = \big( \{ (x, KU_{T,x}) \}_{x
    \in \tb{KUNodes}(\mc{BT}, RL, T)} \big)$.

\item [\tb{DeriveDK}($SK_{ID}, KU_T, PP$):] It finds a common node $x$ in
    both $SK_{ID}$ and $KU_T$. If it fails to find, then it returns
    $\perp$. Note that If $ID$ was not revoked during the time period $T$,
    then there exist a node $x \in \tb{Path}(\mc{BT}, x_{ID}) \cap
    \tb{KUNodes}(\mc{BT}, RL, T)$.
    For this node $x$, it retrieves $SK_{ID,x} = (K_{x,0}, K_{x,1})$ and
    $KU_{T,x} = (U_{x,0}, U_{x,1})$ from $SK_{ID}$ and $KU_T$ respectively.
    It chooses random $r_0, r_1 \in \Z_p^*$ and outputs a decryption key
    $DK_{ID,T} = \big( D_1 = K_{x,0} \cdot U_{x,0} \cdot F_u(ID)^{r_0}
    \cdot F_h(T)^{r_1}, D_2 = K_{x,1} \cdot g^{r_0}, D_3 = U_{x,1} \cdot
    g^{r_1} \big)$.

\item [\tb{Encrypt}($ID, T, M, PP$):] Let $\mc{ET}$ be a binary tree for
    time periods and $v_T$ be a leaf node associated with $T$ in $\mc{ET}$.
    It chooses random $s \in \Z_p^*$ and computes $C_0 = e(g_1, g_2)^s
    \cdot M, C_1 = g^{-s}, C_2 = F_u(ID)^s$.
    For each node $v \in \tb{CTNodes}(\mc{ET}, T)$, it performs: 1) It
    chooses random $s_v \in \Z_p^*$ and sets $s_v = s$ if $v = v_T$. 2) It
    calculates $CT_v = \big( C_{v,0} = \big( h_0 \prod_{j=1}^{|b_v|}
    h_j^{b_v[j]} \big)^{s_v}, C_{v,|b_v|+1} = h_{|b_v|+1}^{s_v}, \ldots,
    C_{v,\ell} = h_{\ell}^{s_v} \big)$.
    Finally, it outputs a ciphertext $CT_{ID,T} = \big( ID, T, C_0, C_1,
    C_2, \{ CT_v \}_{v \in \tb{CTNodes}(\mc{ET}, T)} \big)$.

\item [\tb{UpdateCT}($CT_{ID,T}, T', PP$):] Let $v_{T}, v_{T'}$ be leaf
    nodes in $\mc{ET}$ assigned to $T, T'$, respectively. If $T' < T$, then
    it returns $\perp$ to indicate that $T'$ is invalid.
    It chooses random $s' \in \Z_p^*$ and computes $C'_0 = C_0 \cdot e(g_1,
    g_2)^{s'}, C'_1 = C_1 \cdot g^{-s'}, C'_2 = C_2 \cdot F_u(ID)^{s'}$.
    For each node $v' \in \tb{CTNodes}(\mc{ET}, T')$, it performs: 1) It
    find a node $v \in \tb{CTNodes}(\mc{ET}, T)$ such that $b_v$ is a
    prefix of $b_{v'}$. 2) It chooses random $s_{v'} \in \Z_p^*$ and sets
    $s_{v'} = s'$ if $v' = v_{T'}$. 3) It calculates $CT_{v'} = \big(
    C_{v',0} = C_{v,0} \cdot \prod_{j=|b_v|+1}^{|b_{v'}|} C_{v,j} \cdot
    \big( h_0 \prod_{j=1}^{|b_{v'}|} h_j^{b_{v'}[j]} \big)^{s_{v'}},
    C_{v',|b_{v'}|+1} = C_{v,|b_{v'}|+1} \cdot h_{|b_{v'}|+1}^{s_{v'}},
    \ldots, C_{v',|b_{v'}|+\ell} = C_{v,\ell} \cdot h_{\ell}^{s_{v'}}
    \big)$.
    Finally, it outputs an updated ciphertext $CT_{ID, T'} = \big( ID, T',
    \lb C'_0, C'_1, C'_2, \{ CT_{v'} \}_{v' \in \tb{CTNodes}(\mc{ET},
    T')}\big)$.

\item [\tb{Decrypt}($CT_{ID,T}, DK_{ID,T'}, PP$):] Let $DK_{ID,T'} = ( D_1,
    D_2, D_3 )$.
    If $T' < T$, then it returns $\perp$. Otherwise, it updates $CT_{ID,T}$
    to obtains $CT_{ID,T'} = ( ID, T', C'_0, C'_1, C'_2, \{ CT_{v'} \}_{v'
    \in \tb{CTNodes}(\mc{ET}, T')} )$ where $CT_{v'} = (C_{v',0}, \ldots,
    C_{v',\ell})$ by running $\tb{UpdateCT}(CT_{ID,T}, T', PP)$.
    It outputs a message $M$ by computing $C'_0 \cdot e(C'_1, D_1) \cdot
    e(C'_2, D_2) \cdot e(C_{v_{T'},0}, D_3)$ where $v_{T'}$ is a leaf node
    associated with $T'$.

\item [\tb{Revoke}($ID, T, RL, ST$):] It adds $(ID, T)$ to $RL$ and returns
    the updated $RL$.
\end{description}

Wei et al. claimed that above RS-IBE scheme is correct and secure if the
$\ell$-BDHE assumption holds.

\section{Analysis of the RS-IBE Scheme}

In this section, we show that the above RS-IBE scheme is not correct since
the decryption fails if the ciphertext time $T$ is less than the decryption
key time $T'$.

\begin{lemma}  \label{lem:tree-node}
Let $\mc{ET}$ be a binary tree for time periods and $v_T, v_{T'}$ be leaf
nodes associated with time $T, T'$, respectively. If $T + 1 \leq T'$, then
there exists a node $\tilde{v} = \tb{CTNodes}(\mc{ET}, T) \cap
\tb{Path}(v_{T'})$ but $v_T \neq \tilde{v}$. That is, $v_T, \tilde{v} \in
\tb{CTNodes}(\mc{ET}, T)$, $v$ is an ancestor node of $v_{T'}$, and $v_T \neq
\tilde{v}$.
\end{lemma}

\begin{proof}
By the main property of \tb{CTNodes},  we have that for any node $v' \in
\tb{CTNodes}(\mc{ET}, T')$ there is a common node $v''$ such that $v'' =
\tb{CTNodes}(\mc{ET}, T) \cap \tb{Path}(\mc{ET}, v')$ if $T+1 \leq T'$.
Therefore, for both nodes $v_T$ and $v_{T'}$ associated with time $T$ and
$T'$, there exists a node $\tilde{v} = \tb{CTNodes}(\mc{ET}, T) \cap
\tb{Path}(\mc{ET}, v_{T'}$. Now, Let's show that the node $v_T$ and the node
$\tilde{v}$ are different. In the given condition, $T+1 \leq T'$ is
established, and each time is sequentially assigned to a leaf node in
$\mc{ET}$. Therefore, two nodes $v_T, v_{T'}$ are different nodes since $T
\neq T'$ and they are assigned to leaf nodes. Since the node $\tilde{v}$
belongs to the path nodes $\tb{Path}(\mc{ET}, v_{T'})$, the node $\tilde{v}$
can never be a leaf node if $\tilde {v} \neq v_{T'}$. Therefore, $\tilde{v}
\neq v_T$ is established, since $v_T \neq v_{T'}$ and $v_T$ is a leaf node.
\end{proof}

\begin{theorem}
Let $CT_{ID,T}$ be a ciphertext associated with time $T$ and $DK_{ID, T'}$ be
a decryption key associated with time $T'$. If $T + 1 \leq T'$, then the
ciphertext $CT_{ID,T}$ cannot be decrypted by using the decryption key
$DK_{ID, T'}$ in the decryption algorithm.
\end{theorem}

\begin{proof}
To prove this theorem, we first analyze nodes in the binary tree $\mc{ET}$
which are associated with the ciphertext elements used in the decryption
algorithm and then analyze how the random exponents of these ciphertext
elements are constructed. After that, we argue that the decryption algorithm
will fail due to the random exponents of the ciphertext elements which are
used for decryption.

The decryption algorithm takes an original ciphertext $CT_{ID,T}$ and a
decryption key $DK_{ID,T'}$ as input. Then, it performs the \tb{UpdateCT}
algorithm to derive an updated ciphertext $CT_{ID,T'}$ since $T < T'$ . Next,
it uses the updated ciphertext element $C_{v_{T'},0}$, which is related to a
leaf node $v_{T'}$ associated with the time $T'$, for the decryption. Here,
the \tb{UpdateCT} algorithm finds the node $\tilde{v}$ which is an ancestor
node of $v_{T'}$ and belongs to the set $\tb{CTNodes}(\mc{ET}, T)$, and
delegates the ciphertext element $C_{\tilde{v},0}$ to obtain the ciphertext
element $C_{v_{T'},0}$. From the Lemma \ref{lem:tree-node}, we have that the
node $\tilde{v}$ which belongs to $\tb{CTNodes}(\mc{ET}, T)$ is not equal to
the node $v_T$ if $T+1 \leq T'$.

Now, we analyze random exponents in the original ciphertext $CT_{ID,T}$ which
are associated with the nodes in $\tb{CTNodes}(\mc{ET},T)$. The encryption
algorithm generates ciphertext elements for nodes in $\tb{CTNodes}(\mc{ET},
T)$. According to the encryption algorithm, for each node $v \in
\tb{CTNodes}(\mc{ET}, T)$, if $v = v_T$, then the same random exponent $s$
which is used for message encryption is used to generate $C_{v,0}$. If $v
\neq v_T$, then a new random exponent $s_v$ is selected to generate
$C_{v,0}$. However, since $\tilde{v} \neq v_T$ from the previous Lemma
\ref{lem:tree-node}, the random exponent $s_{\tilde{v}}$ is not equal to $s$
with high probability where $s_{\tilde{v}}$ is used for the node $\tilde{v}$.

The decryption algorithm finally calculates the following equation by using
the ciphertext elements and decryption elements.
    \begin{align*}
    &   C'_0 \cdot e(C'_1, D_1) \cdot e(C'_2, D_2) \cdot e(C_{v_{T'},0}, D_3) \\
    &=  M \cdot e(g_1, g_2)^s \cdot
        e(g^{-s}, g_2^{\alpha} F_u(ID)^r_0 F_h(T)^{r_1}) \cdot
        e(F_u(ID)^s, g^{r_0}) \cdot
        e(F_h(T)^{s_{\tilde{v}}}, g^{r_1}) \\
    &=  M \cdot
        e(g^{-s}, F_h(T)^{r_1}) \cdot
        e(F_h(T)^{s_{\tilde{v}}}, g^{r_1}) \\
    &=  M \cdot e(F_h(T), g)^{(s_{\tilde{v}} - s) r_1}.
    \end{align*}
Note that we ignored the re-randomization process since it does not affect
our analysis. In order to correctly obtain the message contained in the
ciphertext, it is required that $(s_{\tilde{v}} - s) \equiv 0 \mod p$ should
be satisfied. However, in the previous analysis, this relation cannot be
satisfied because the ciphertext element associated with the node $\tilde{v}$
of the original ciphertext uses a new random exponent $s_{\tilde{v}}$. Thus,
the decryption can be successful if $T = T'$, but the decryption always fails
except with negligible probability if $T+1 \leq T'$.
\end{proof}

\section{Modification to the RS-IBE Scheme}

In the previous section, we have shown that the RS-IBE scheme of Wei et al.
does not satisfy the correctness, which is the minimum requirement that the
cryptographic scheme must satisfy, due to the problem of random exponents in
the encryption algorithm. In this section, we propose a modification to the
RS-IBE scheme of Wei et al. to guarantee the correctness and the security.

A simple way to modify the RS-IBE scheme of Wei et al. \cite{WeiLH18} is to
force the ciphertext elements associated with $\tb{CTNodes}(\mc{ET}, T)$ to
use the same random exponent $s$ which is used to encrypt a message in the
ciphertext component $C_0$. In this case, there is no problem such that the
decryption algorithm fails when the ciphertext is updated since $s_v = s$ for
all nodes $v$. However, this simple modification does not lead to a secure
RS-IBE scheme. The reason is that if multiple nodes are provided with
ciphertext elements $\{ h_j^s \}$ associated with the same random $s$, it is
possible for anyone to use these elements to modify the original ciphertext
element with current time to derive another ciphertext element with past
time. This makes it possible for a revoked user to modify the ciphertext with
current time to obtain a ciphertext with past time to decrypt the original
ciphertext. Therefore, this simple method does not work.

A secure and efficient method to modify the RS-IBE scheme is to use a
cryptographic scheme that supports ciphertext update functionality. Lee et
al. \cite{LeeCLPY13,LeeCLPY17,Lee16} introduced the concept of self-updatable
encryption and proposed secure SUE schemes that efficiently handle ciphertext
updates. Thus, we can modify the RS-IBE scheme of Wei et al. to use the SUE
scheme for the ciphertext update components and key update components. The
secure SUE scheme proposed by Lee et al. supports correct decryption although
it uses different random exponents in ciphertext elements associated with
$\tb{CTNodes}(\mc{ET}, T)$. Additionally, this modified RS-IBE scheme can
reduce the number of ciphertext elements from $O(\log^2 T_{max})$ to $O(\log
T_{max})$ because of the efficiency of the SUE scheme. We also note that  the
existing RS-ABE scheme can be easily converted to an RS-IBE scheme by
changing the attribute set of ABE to the identity of IBE.

\section{Conclusion}

In this paper, we pointed out that the RS-IBE scheme of Wei et al. does not
provide the correctness property. The problem of the RS-IBE scheme was that
when a ciphertext with time $T$ is updated to a ciphertext with time $T+1$,
this updated ciphertext cannot be decrypted by using a decryption key with
time $T+1$. The main reason of this problem was that the random exponent of
the ciphertext element associated with a tree node corresponding to time
$T+1$ was not the same as the random exponent used to encrypt a message in
the ciphertext. This decryption problem cannot be solved in a simple way, so
we proposed a method to modify the previous RS-IBE scheme to be a secure and
efficient RS-IBE scheme using a self-updatable encryption scheme.

\bibliographystyle{plain}
\bibliography{analysis-of-rsibe}

\end{document}